\documentclass[11pt,a4paper,reqno,draft]{amsart}
\usepackage[T1]{fontenc}
\usepackage{mathrsfs}
\usepackage{amsfonts,amssymb,amsmath,amsgen,amsthm}
\usepackage{color}
\usepackage{amsbsy}
\usepackage{mathrsfs}
\usepackage{bbm}
\usepackage{titletoc}
\usepackage{enumerate}
\usepackage{hyperref}
\usepackage{subfigure}
\usepackage[subfigure]{ccaption}
\usepackage[all]{xy}
\theoremstyle{plain}
\newtheorem{theorem}{Theorem}[section]

\newtheorem{lemma}[theorem]{Lemma}
\newtheorem{corollary}[theorem]{Corollary}
\newtheorem{proposition}[theorem]{Proposition}

\newtheorem{remark}[theorem]{Remark}

\newtheorem{example}[theorem]{Example}

\catcode`@=11
\def\section{\@startsection{section}{1}%
  \z@{1.5\linespacing\@plus\linespacing}{.5\linespacing}%
  {\normalfont\bfseries\large\centering}}
\catcode`@=12
\def\CC{{\mathbb C}}
\def\RR{{\mathbb R}}
\def\NN{{\mathbb N}}
\def\calO{\mathcal O}

\def\({\left(}
\def\){\right)}
\def\<{\left\langle}
\def\>{\right\rangle}

\def\eps{\varepsilon}

\DeclareMathOperator{\RE}{Re}

\numberwithin{equation}{section}

\newcommand{\be}{\begin{equation}}
\newcommand{\ee}{\end{equation}}
\newcommand{\bea}{\begin{eqnarray}}
\newcommand{\eea}{\end{eqnarray}}
\newcommand{\bee}{\begin{eqnarray*}}
\newcommand{\eee}{\end{eqnarray*}}
\def\ds{\displaystyle}
\def\ni{\noindent}
\def\bs{\bigskip}
\def\ms{\medskip}

\def\eps{\varepsilon}

\def\fref#1{{\rm (\ref{#1})}}
\def\pref#1{{\rm \ref{#1}}}

\def\pa{\partial}
\def\na{\nabla}

\def\calH{{\mathcal H}}

\def\bx{{\bold x}}

\def\psie{{\Psi^{\eps,\alpha}}}

\def\beq{{A^{\eps,\alpha}}}
\def\beqo{{A^{0,\alpha}}}
\def\beqs{{A^{\eps,0}}}
\def\beqoo{{A}}

\begin{document}

\title[Dimension reduction for GPE]{Dimension reduction for anisotropic Bose-Einstein condensates in the strong interaction regime}
\author[W. Bao]{Weizhu Bao}
\email{matbaowz@nus.edu.sg}
\address{Department of Mathematics, National University of Singapore, Singapore 119076, Singapore}
\author[L. Le Treust]{Lo\"{i}c Le Treust}
\email{loic.letreust@univ-rennes1.fr}
\address{IRMAR, Universit\'e de Rennes 1 and INRIA, IPSO Project}
\author[F. M\'ehats]{Florian M\'ehats}
\email{florian.mehats@univ-rennes1.fr}
\address{IRMAR, Universit\'e de Rennes 1 and INRIA, IPSO Project}

\begin{abstract}
We study the problem of dimension reduction for the three dimensional Gross-Pitaevskii equation (GPE) describing a Bose-Einstein condensate confined in a strongly anisotropic harmonic trap. Since the gas is assumed to be in a strong interaction regime, we have to analyze two combined singular limits: a semi-classical limit in the transport direction and the strong partial confinement limit in the transversal direction. We prove that both limits commute together and we provide convergence rates. The by-products of this work are approximated models in reduced dimension for the GPE, with a priori estimates of the approximation errors.
\end{abstract}
\maketitle

\section{Introduction and main results}
\label{sec:intro}

In this paper, we study dimension reduction for the three-dimensional Gross-Pitaevskii equation (GPE) modeling Bose-Einstein condensation \cite{Anderson,Bradley,Davis}. In contrast with the existing literature on this topic \cite{ben2005nonlinear,BenAbdallah2008154,abdallah2011second}, we will {\em not} assume that the gas is in a weak interaction regime.

Based on the mean field approximation \cite{LiebSeiringerPra2000,LiebSeiringer,Erdos},
the Bose-Einstein condensate is modeled by its wavefunction 
$\Psi:=\Psi(t,\bx)$ satisfying the GPE written in physical variables as
\be
\label{GPE1}
i\hbar \pa_t \Psi=-\frac{\hbar^2}{2m}\Delta \Psi+V(\bx)\Psi+Ng|\Psi|^2\Psi,
\ee
where $\Delta$ is the Laplace operator, $V(\bx)$ denotes the trapping harmonic potential, $m>0$ is the mass, $\hbar$ is the Planck constant, $g=\frac{4\pi\hbar^2a_s}{m}$ describes the interaction between atoms in the condensate with the $s$-wave scattering length $a_s$ and $N$ denotes the number of particules in the condensate. The wave function is normalized according to
$$\int |\Psi(t,\bx)|^2d\bx=1.$$

\subsection{Scaling assumptions}
We assume that the harmonic potential is strongly anisotropic and confines particles from dimension $n+d$ to dimension $n$. In applications, we will have $n+d=3$ and, either $n=2$ for disk-shaped condensates, or $n=1$ for cigar-shaped condensates. We shall denote $\bx=(x,z)$, where $x\in \RR^n$ denotes the variable in the confined direction(s) and $z\in \RR^d$ denotes the variable in the transversal direction(s). The harmonic potential reads \cite{Baocai2013,Pethick,PitaevskiiStringari}
$$V(\bx)=\frac{m}{2}\left({\omega}_x^2|x|^2+{\omega}_z^2|z|^2\right)$$
where ${\omega}_z\gg{\omega}_x$. We introduce the two dimensionless parameters
$$\eps=\sqrt{{\omega}_x/{{\omega}_z}},\qquad \beta=\frac{4\pi Na_s}{a_0},$$
where the harmonic oscillator length is defined by \cite{Baocai2013,Pethick,PitaevskiiStringari}
$$a_0=\(\frac{\hbar}{m{\omega}_x}\)^{1/2}.$$

Let us rewrite the GPE \fref{GPE1} in dimensionless form. For that, we introduce the new variables $\tilde t$, $\tilde x$, $\tilde z$ and the associated unknown $\widetilde \Psi$ defined by \cite{Baocai2013,Pethick,PitaevskiiStringari}
$$\tilde t={\omega}_x t,\qquad \tilde x=\frac{x}{a_0},\qquad \tilde z=\frac{z}{a_0},\qquad \widetilde \Psi(\tilde t,\tilde x,\tilde z)=a_0^{(n+d)/2}\Psi(t,x,z).$$
The dimensionless GPE equation reads \cite{Baocai2013,Pethick,PitaevskiiStringari}
\begin{equation}\label{GPE0}
i\pa_{\tilde t} \widetilde \Psi=-\frac{1}{2}\Delta \widetilde \Psi+\frac{1}{2}\left(|\tilde x|^2+\frac{1}{\eps^4}| \tilde z|^2\right)\widetilde \Psi+\beta|\widetilde \Psi|^2\widetilde \Psi.
\end{equation}
In order to observe the condensate at the correct space scales, we will now proceed to a rescaling in $x$ and $z$. Let us define
\begin{equation}
\label{alpha}
	\alpha=\eps^{2d/n}\beta^{-2/n},
\end{equation}
 and set
\[
t'=\tilde t,\qquad z'=\frac{\tilde z}{\eps},\qquad x'=\alpha^{1/2}\tilde x,
\]
which means that the typical length scales of the dimensionless variables are $\eps$ in the $z$-direction and $\alpha^{-1/2}$ in the $x$-direction. The wavefunction is rescaled as follows:
$$\psie(t',x',z'):=\eps^{d/2}\alpha^{-n/4}\widetilde \Psi(\tilde t,\tilde x,\tilde z)e^{i\tilde td/2\eps^2}.$$
Notice that the $L^2$ norm of $\psie$ is left invariant by this rescaling, so we still have
$$\int_{\RR^{n+d}} |\Psi^{\eps,\alpha}(t,x,z)|^2dxdz=1.$$
%
We end up with the following rescaled GPE (for simplicity we omit the primes on the variables):
\begin{equation}\label{gpe}
i\alpha \pa_t \psie=\frac{\alpha}{\eps^2}\calH_z\psie-\frac{\alpha^2}{2}\Delta_x \psie+\frac{|x|^2}{2}\psie+\alpha|\psie|^2\psie
\end{equation}
where the transversal Hamiltonian is $$\calH_z:=-\frac{1}{2}\Delta_z+\frac{|z|^2}{2}-\frac{d}{2}$$ and the scaling assumptions are $$\alpha\ll 1\quad\mbox{and}\quad \eps\ll 1.$$ The spectrum of $\calH_z$ is the set of integers $\NN$, its ground state (associated to the eigenvalue 0) is ${\omega_0}(z)=\pi^{-d/4}e^{-|z|^2/2}$.

\bs
The dimension reduction of the GPE (\ref{GPE0}) from three dimensions (3D) to lower dimensions was studied formally
in \cite{BaoJakschP,BaoP} and numerically in \cite{BaoGeJakschPW} for fixed $\beta$ when $\varepsilon\to0$. 
The mathematical rigorous justification for this dimension reduction was given 
in \cite{ben2005nonlinear,BenAbdallah2008154,abdallah2011second,BaoBenCai} for $\alpha=1$  and $\eps\ll1$ in 
 (\ref{gpe}) so that $\beta= \eps^d\ll1$ in (\ref{GPE0}), which corresponds to a weak interaction regime
in the GPE (\ref{GPE0}). However, it is an open problem to justify mathematically the dimension reduction
of the GPE (\ref{GPE0}) in the strong interaction regime, i.e. for fixed $\beta$ when $\varepsilon\to0$.
The key difficulty is due to that the energy associated to the reduced GPE in lower dimensions is unbounded
when $\varepsilon \to0$ \cite{BaoJakschP,Baocai2013}. In this paper, 
we study the strong interaction regime by adapting
a proper re-scaling. This amounts to considering simultaneously the strong confinement limit and the semi-classical limit for the solution $\psie$ to \fref{gpe} as $\eps\to 0$ and $\alpha\to 0$. Note that $\beta= \eps^d\alpha^{-n/2}$ may tend to every constants $\gamma\in \RR^+$ and even to $+\infty$.
	
	Our key mathematical assumption will be that the wavefunction $\psie$ at time $t=0$ is under the WKB form:
	\begin{equation}
	\label{initialWKB}
		\psie(0,x,z)= \Psi^\alpha_{\rm init}(x,z):=A_0(x,z)e^{iS_0(x)/\alpha},\quad \forall (x,z)\in \RR^{n+d}.
	\end{equation}
Here $A_0$ is a complex-valued function and $S_0$ is real-valued.
\begin{remark}
\label{rem1}
With respect to the small parameter $\alpha$, Eq. \eqref{gpe} is in a semiclassical regime which is usually referred to as "weakly nonlinear geometric optics", see \cite{bookRemi}. The more singular regime
\begin{equation}
\label{strong}
i\alpha \pa_t \Psi=\frac{\alpha}{\eps^2}\calH_z\Psi-\frac{\alpha^2}{2}\Delta_x \Psi+\frac{|x|^2}{2}\Psi+|\Psi|^2\Psi
\end{equation}
would correspond to the choice
$$\alpha=\left(\eps^d \beta^{-1}\right)^{\frac{2}{n+2}},$$
instead of the choice $\alpha=\left(\eps^d \beta^{-1}\right)^{\frac{2}{n}}$ that we have made in \eqref{alpha}. Hence, the difference between these two regimes lies in the assumption on the initial wavefunction: in the regime \eqref{gpe} studied here, the wavefunction is assumed to have a broader extension in the $x$ direction than in the more singular regime \eqref{strong}.
\end{remark}
%
\subsection{Heuristics} In the section, we derive formally the limiting behavior of the solution of \eqref{gpe}. We have the choice to first let $\eps\to 0$ (strong confinement limit), then $\alpha\to 0$ (semiclassical limit), or to exchange these two limits: first  $\alpha\to 0$, then  $\eps\to 0$. Our main result, stated in the next section, will be that in fact both limits commute together: the limit is valid as $\eps$ and $\alpha$ converge {\sl independently} to zero.

\subsubsection*{a) Strong confinement limit first, then semiclassical limit}
Following \cite{BenAbdallah2008154}, in order to analyze the strong partial confinement limit, it is convenient to begin by filtering out the fast oscillations at scale $\eps^2$ induced by the transveral Hamiltonian. To this aim, we introduce the new unknown
	\[
		\Phi^{\eps,\alpha}(t,\cdot)=e^{it\mathcal{H}_z/\eps^2}\psie(t,\cdot).
	\]
It satisfies the equation
	\bee
		i\alpha \pa_t \Phi^{\eps,\alpha}=-\frac{\alpha^2}{2}\Delta_x \Phi^{\eps,\alpha}+\frac{|x|^2}{2}\Phi^{\eps,\alpha}+\alpha F\left(\frac{t}{\eps^2},\Phi^{\eps,\alpha}\right)
	\eee
	where the nonlinear function is defined by
	\begin{equation}
	\label{def:F}
		\begin{array}{lll}
			 F(\theta,\Phi)=e^{i\theta\mathcal{H}_z}\left(\left|e^{-i\theta\mathcal{H}_z}\Phi\right|^2e^{-i\theta\mathcal{H}_z}\Phi\right).
		\end{array}
	\end{equation}
	A fundamental remark is that for all fixed $\Phi$, the function $\theta\mapsto F(\theta, \Phi)$ is $2\pi$-periodic, since the spectrum of $\mathcal{H}_z$ only contains integers.
	For any fixed $\alpha>0$, Ben Abdallah et al. \cite{BenAbdallah2008154,abdallah2011second} proved by an averaging argument that we have $\Phi^{\eps,\alpha}=\Phi^{0,\alpha}+\mathcal O(\eps^2)$, where $\Phi^{0,\alpha}$ solves the averaged equation
	\begin{equation}\label{firststep}
		i\alpha \pa_t \Phi^{0,\alpha}=-\frac{\alpha^2}{2}\Delta_x \Phi^{0,\alpha}+\frac{|x|^2}{2}\Phi^{0,\alpha}+\alpha F_{av}(\Phi^{0,\alpha}), \qquad \Phi^{0,\alpha}(t=0)=\Psi^\alpha_{\rm init}
	\end{equation}
where $F_{av}$ is the averaged vector field
	\begin{equation}
	\label{def:Fav}
			F_{av}(\Phi)=\frac{1}{2\pi}\int_0^{2\pi}F(\theta,\Phi)d\theta.
	\end{equation}
Now we can proceed to the second limit $\alpha\to 0$. As we said in Remark \ref{rem1}, \eqref{firststep} is written in the semi-classical regime of "weakly nonlinear geometric optics", which can be studied by a WKB analysis. Here we are only interested in the limiting model, so in the first stage of the WKB expansion. Let us introduce the solution $S(t,x)$ of the eikonal equation
\begin{equation}
		\pa_t S +\frac{|\nabla_x S|^2}{2}+\frac{|x|^2}{2} = 0,\qquad S(0,x)=S_0(x)\label{eq:S}
		\end{equation}
and, again, filter out the oscillatory phase of the wavefunction by setting
\[
	\beqo=e^{-iS(t,x)/\alpha}\,\Phi^{0,\alpha}.
\]
This function $\beqo(t,x,z)$ satisfies
	\begin{equation}\label{eq:b2}
					\pa_t \beqo +\nabla_x S \cdot \nabla_x \beqo  +\frac{1}{2} \beqo \Delta_x S=i\frac{\alpha}{2}\Delta_x \beqo -iF_{av}(\beqo),
\end{equation}	
with the initial data
$$\beqo(0,x,z)=A_0(x,z).$$
As long as the phase $S(t,x)$ remains smooth, i.e. before the formation of caustics in the eikonal equation, we expect to have $\beqo=\beqoo+\mathcal O(\alpha)$,  where $\beqoo(t,x,z)$ solves the limiting transport equation
			\begin{equation}\label{eq:b3}
								\pa_t \beqoo +\nabla_x S \cdot \nabla_x \beqoo  +\frac{1}{2}\beqoo \Delta_x S= -iF_{av}(\beqoo),\qquad \beqoo(0,x,z)=A_0(x,z).
			\end{equation}
To summarize, the solution $\Psi^{\eps,\alpha}$ of \eqref{gpe} is expected to behave as
\begin{equation}
\label{estimate}\Psi^{\eps,\alpha}(t,x,z)=e^{-it\mathcal{H}_z/\eps^2}e^{iS(t,x)/\alpha}\,A(t,x,z)+\mathcal O(\eps^2)+\mathcal O(\alpha).
\end{equation}

\subsubsection*{b) Semiclassical limit first, then strong confinement limit}
Coming back to the GPE \eqref{gpe}, let us first proceed to the semiclassical limit $\alpha\to 0$. We define
\begin{equation}
\label{beq}
	\beq =e^{it\mathcal{H}_z/\eps^2}e^{-iS(t,x)/\alpha}\,\Psi^{\eps,\alpha},
\end{equation}
where $S(t,x)$ is still the solution of the eikonal equation \eqref{eq:S}.
A direct computation shows that this function satisfies the equation
\begin{align}\label{eq:b1}
		\pa_t \beq +\nabla_x S \cdot \nabla_x \beq  +\frac{1}{2} \beq \Delta_x S&=i\frac{\alpha}{2}\Delta_x \beq-iF\left(\frac{t}{\eps^2},\beq\right),\\
	\nonumber\beq(0,x,z)&=A_0(x,z),
\end{align}
where $F$ is still defined by \eqref{def:F}. For all fixed $\eps$, we can expect that, as $\alpha\to 0$, we have $\beq=\beqs+\mathcal O(\alpha)$, where $\beqs$ solves the equation
\begin{align}\label{eq:b6}
		\pa_t \beqs +\nabla_x S \cdot \nabla_x \beqs  +\frac{1}{2} \beqs \Delta_x S&=-iF\left(\frac{t}{\eps^2},\beqs\right),\\
	\nonumber\beqs(0,x,z)&=A_0(x,z).
\end{align}
The last step consists in letting $\eps\to 0$ in this equation (strong confinement limit), which amounts to average out the oscillatory nonlinear term in \eqref{eq:b6}. This step yields the limiting equation \eqref{eq:b3}, and we have $\beqs=A+\mathcal O(\eps^2)$.
\begin{remark}		
	A key point here in this analysis is that the nonlinearities $F$ and $F_{av}$ are gauge invariant \textit{i.e.} for all $U\in L^2(\RR^{n+d})$ and for all $t$, we have
	\[
		F(t,Ue^{iS/\alpha}) = F(t,U)e^{iS/\alpha}, \qquad
		F_{av}(Ue^{iS/\alpha})= F_{av}(U)e^{iS/\alpha}.
	\]
\end{remark}

%
%
%
\subsection{Main results}
	Our main contribution is to prove rigorously the limit of the coupled averaging and semi-classical limits as $\eps\to 0$ and $\alpha\to 0$ independently and to prove the estimate \eqref{estimate}. It is natural --\,and equivalent as long as the phase $S(t,x)$ is well defined and is smooth\,-- to work with the function $\beq$ defined by \eqref{beq}.
\subsubsection{Existence, uniqueness and uniform boundedness results}
	Let us make precise our functional framework. For wavefunctions, we will use the scale of Sobolev spaces adapted to quantum harmonic oscillators:
\[
	B^m(\RR^{n+d}):=\{u\in H^m(\RR^{n+d})\;\mbox{such that}\,\left(|x|^m+|z|^m\right)u\in L^2(\RR^{n+d})\}
\]
for $m\in \NN$. For the phase $S$, we will use the space of subquadratic functions, defined by
\begin{equation*}
  {\tt SQ}_k(\RR^n) = \{f\in \mathcal C^k(\RR^n;\RR))\;\mbox{such that}\;\partial^\kappa_x f\in L^\infty(\RR^n), \,\mbox{for all}\; 2\leq |\kappa|\leq k\}.
\end{equation*}
where $k\in\NN$, $k\geq 2$.
In the following theorem, we give existence and uniqueness results for equations \eqref{eq:S}, \eqref{eq:b2}, \eqref{eq:b3}, \eqref{eq:b1} and \eqref{eq:b6} as well as uniform bounds on the solutions.
\begin{theorem}\label{theo:local1}
	Let $A_0\in B^m(\RR^{n+d})$ and $S_0\in {\tt SQ}_{s+1}(\RR^n)$, where $m>\frac{n+d}{2}+2$ and $s\geq m+2$.
 Then the following holds:
  \begin{enumerate}[(i)]
  \item There exists $T>0$ such that the eikonal equation \fref{eq:S} admits a unique solution
  $S\in \mathcal{C}([0,T];{\tt SQ}_s(\RR^n))\cap \mathcal C^s([0,T]\times \RR^n)$.
  \item There exists $\overline T\in(0,T]$ independent of $\eps$ and $\alpha$ such that the solutions $\beq$, $\beqo$, $\beqs$ and $\beqoo$ of, respectively, \eqref{eq:b1}, \eqref{eq:b2}, \eqref{eq:b6} and \eqref{eq:b3}, are uniquely defined in the space
  \[
  	C([0, \overline T];B^m(\RR^{n+d}))\cap C^1([0,\overline T];B^{m-2}(\RR^{n+d})).
\]
  \item  For all $(\eps,\alpha)\in (0,1]^2$, the functions $\beq, \beqo,\beqs, \beqoo$ are uniformly  bounded in
  	\[
		C([0,\overline T];B^m(\RR^{n+d}))\cap C^1([0,\overline T];B^{m-2}(\RR^{n+d})).
	\]
  \end{enumerate}
\end{theorem}
%
%
%
\subsubsection{Study of the limits $\alpha\to0$ and $\eps\to 0$.}
We are now able to study the behavior of $\beq$ as $\eps\to 0$ and $\alpha\to 0$.
	%
	%
	%
	\begin{theorem}\label{theo:limit}
		Assume the hypothesis of Theorem \pref{theo:local1} true. Then, for all $(\eps,\alpha)\in (0,1]^2$, we have the following bounds:
		\begin{enumerate}[(i)]
			\item \label{theol:pt1}Averaging results:
			\begin{equation}
			\label{error1}
				\|\beq-\beqo\|_{L^\infty([0,\overline T];B^{m-2}(\RR^{n+d}))}\leq C\eps^{2}
			\end{equation}
			and
			\begin{equation}
			\label{error2}
				\|\beqs-\beqoo\|_{L^\infty([0,\overline T];B^{m-2}(\RR^{n+d}))}\leq C\eps^2.
			\end{equation}
			\item Semi-classical result:
			\begin{equation}
			\label{error3}
				\|\beq-\beqs\|_{L^\infty([0,\overline T];B^{m-2}(\RR^{n+d}))}\leq C\alpha
			\end{equation}
			and
			\begin{equation}
			\label{error4}
				\|\beqo-\beqoo\|_{L^\infty([0,\overline T];B^{m-2}(\RR^{n+d}))}\leq C\alpha.
			\end{equation}
			\item Global result:
			\begin{equation}
			\label{error5}
\|\beq-\beqoo\|_{L^\infty([0,\overline T];B^{m-2}(\RR^{n+d}))}\leq C(\eps^2+\alpha).
			\end{equation}
		\end{enumerate}
		The constant $C$ here does not depend on $\alpha$ and $\eps$.
These estimates can be summarized in the following diagram:
	\[
\xymatrix{
    \beq \ar[rr]^{\calO(\eps^{2})} \ar[dd]_{\calO(\alpha)} \ar[rrdd]^{\calO(\alpha+\eps^2)} && \beqo \ar[dd]^{\calO(\alpha)}\\
    \\
    \beqs \ar[rr]_{\calO(\eps^2)} && \beqoo
  }
  	\]
\end{theorem}

\bs
	Coming back to the original unknown, our theorem can be expressed in terms of Gross-Pitaevskii equations.
		\begin{corollary}\label{cor:schro}
		Under the assumptions  of Theorem \pref{theo:local1}, the solution $\Psi^{\eps,\alpha}$ of \eqref{gpe}, \eqref{initialWKB} satisfies
		\[
			\left\|e^{-iS/\alpha}\left(\psie -e^{-it\mathcal H_z/\eps^2}\Phi^{\alpha}\right)\right\|_{L^\infty([0,\overline T];B^{m-2}(\RR^{n+d}))}\leq C\eps^{2},
		\]
		 where $\Phi^\alpha$ solves the Gross-Pitaevskii equation in reduced dimension \eqref{firststep}. Moreover, we have
		\[
			\left\|e^{-iS/\alpha}\,\Phi^\alpha -\beqoo \right\|_{L^\infty([0,\overline T];B^{m-2}(\RR^{n+d}))}\leq C\alpha,
		\]
		where $\beqoo $ solves the transport equation \eqref{eq:b3}. The constant $C$ here is independent of $\eps\in (0,1]$ and $\alpha\in (0,1]$.
	\end{corollary}
	\begin{remark}
		For any fixed $\alpha>0$, the previous Corollary implies the result of Ben Abdallah et al. \cite{abdallah2011second,BenAbdallah2008154} with the same rate of convergence $\eps^2$.
	\end{remark}

\begin{remark} An interesting physical case corresponds to initial data polarized on the first eigenmode $\omega_0(z)$ of the confinement Hamiltonian $\mathcal{H}_z$. Assume that the Cauchy condition in \eqref{initialWKB} takes the form
\begin{equation}
	\label{initialWKBpolarized}
		\psie(0,x,z)=\Psi^\alpha_{\rm init}(x,z)=a_0(x)e^{iS_0(x)/\alpha}\omega_0(z).
\end{equation}
Then the solution $\Phi^\alpha$ of the Gross-Pitaevskii equation \eqref{firststep} in reduced dimension remains polarized on $\omega_0$: we have $$\Phi^\alpha(t,x,z)=\varphi^\alpha(t,x)\omega_0(z)$$ and $\varphi^\alpha(t,x)$ solves the equation
$$
i\alpha \pa_t \varphi^{\alpha}=-\frac{\alpha^2}{2}\Delta_x \varphi^{\alpha}+\frac{|x|^2}{2}\varphi^\alpha+\frac{\alpha}{(2\pi)^{d/2}}  |\varphi^{\alpha}|^2\varphi^{\alpha},\qquad \varphi^{\alpha}(0,x)=a_0(x)e^{iS_0(x)/\alpha}.
$$
Moreover, the solution $\beqoo$ of the limiting transport equation \eqref{eq:b3} takes the form $$A(t,x,z)=a(t,x)\omega_0(z),$$ where $a(t,x)$ solves the equation
$$
\pa_t a+\nabla_xS\cdot\nabla_xa+\frac{1}{2}a\Delta_x S=-\frac{i}{(2\pi)^{d/2}} |a|^2a,\qquad a(0,x)=a_0(x).
$$
	\end{remark}
	
	\bs
	The sequel of this article is devoted to the proofs of our two theorems. In Section \ref{sec2}, we prove Theorem \ref{theo:local1} and in Section \ref{sec3}, we prove Theorem \ref{theo:limit}.

	\section{Proof of Theorem \ref{theo:local1}: well-posedness and uniform estimates}
	
This section is devoted to the proof of Theorem \ref{theo:local1}. We first prove the local in time well-posedness of the eikonal equation (Proposition \ref{lem:eik}). Then we prove the local in time well-posedness of the four equations  \eqref{eq:b1}, \eqref{eq:b2}, \eqref{eq:b6} and \eqref{eq:b3}, as well as uniform bounds (Proposition \ref{lem:existbound}). Theorem  \ref{theo:local1} is then a direct consequence of these two Propositions \ref{lem:eik} and \ref{lem:existbound}.

\subsection{Solving the eikonal equation}
\label{subsection:eikonal}
We seek a solution of equation \eqref{eq:S}, where $S_0\in {\tt SQ}_{s+1}(\RR^n)$, for $s\geq 2$.
\begin{example}
	If $S_0=0$, the function defined by
	\[
		S:(t,x)\in (-\pi/2,\pi/2)\times \RR^n\mapsto-\frac{1}{2}|x|^2\tan t
	\]
	is the regular solution of equation \eqref{eq:S}. Let us remark that $S$ is not globally defined in time.
\end{example}
Following \cite{bookRemi}, we use the method of characteristics to find a regular solution to \eqref{eq:S}.
The characteristic equations associated with this Hamilton-Jacobi equation are
\[\left\{\begin{array}{rcll}
\ds	\pa_tx(t,y)&=&\ds \xi(t,y),&\quad x(0,y)=y,\\[2mm]
\ds	\pa_t\xi(t,y)&=&\ds -x(t,y),&\quad \xi(0,y) = \nabla_x S_0(y),\\[2mm]
\ds	\pa_t z(t,y)&=&\ds \frac{|\xi(t,y)|^2}{2}-\frac{|x(t,y)|^2}{2},&\quad z(0,y)= S_0(y)
\end{array}\right.
\]
(see for instance \cite[Section $3.2.5$]{evans1998partial}).

 The two first lines form a closed system of equations which  are called \textit{Hamilton's equations}. The solution is unique, belongs to $\mathcal{C}^{s}(\RR\times \RR^n;\RR^{2n})$ and is given by
\[
	\left(
		\begin{array}{c}
			x(t,y)\\
			\xi(t,y)
		\end{array}
	\right)
	=
	\left(
		\begin{array}{c}
			y\cos(t)+\nabla_x S_0(y)\sin(t)\\
			-y\sin(t)+\nabla_xS_0(y)\cos(t)
		\end{array}
	\right).
\]
Let us define the Jacobian determinant $J\in\mathcal{C}^{s-1}(\RR\times \RR^n;\RR)$ by
\[
	J_t(y) = \det \nabla_yx(t,y)=\det({I}_n\cos(t)+\nabla^2_{xx}S_0(y)\sin(t))
\]
where $I_n$ is the identity matrix of $\RR^n$.
Since $S_0$ is subquadratic, there exists $T>0$ and $C>0$ such that
\[
	\frac{1}{C}<J_t(y)<C,\quad |\pa^\kappa_yx(t,y)|\leq C\mbox{ for all } t\in[0,T],  \; y\in \RR^n \mbox{ and } |\kappa|=1,2.
\]
 By Schwartz's global inversion theorem \cite{schwartz1969nonlinear, derezinski1997scattering}, $y\mapsto x(t,y)$ is a $\mathcal{C}^s$-diffeomorphism of $\RR^n$. Let us denote by $y(t,\cdot)$ its inverse function so that $S$ is defined for all $t\in[0,T]$ and $x\in \RR^n$ by
 \[
 	S(t,x) = z(t,y(t,x)).
 \]
We obtain then the following Proposition, see \cite{evans1998partial} for details:
\begin{proposition}\label{lem:eik}
	If $S_0\in {\tt SQ}_{s+1}(\RR^n)$ with $s\geq 2$, there exists $T>0$ such that the eikonal equation \eqref{eq:S} admits a unique solution $S\in \mathcal C([0,T];{\tt SQ}_s(\RR^n))\cap \mathcal{C}^{s}([0,T]\times\RR^n)$.
\end{proposition}

\subsection{Well-posedness results and uniform estimates} \label{sec2}

Let us introduce the non-negative essentially self-adjoint operator $\Lambda_z^m:=(1+\mathcal{H}_z)^{m/2}$ on $L^2(\RR^{n+d})$ whose domain is $\mathcal{D}(\Lambda_z^m) =  L^2(\RR^n)\otimes B^m_z(\RR^d)$, where
\bee
	 B^m_z(\RR^d):= \{u\in H^m(\RR^d):\quad |z|^{m}u\in L^2(\RR^d)\}.
\eee
The space $B^m(\RR^{n+d})$ is endowed with the norm
\[
	\|u\|_{B^m}^2 := \sum_{|\kappa|\leq m}\|\partial_x^\kappa u\|_{L^2}^2 +\||x|^m u\|_{L^2}^2 +\|\Lambda_z^m u\|_{L^2}^2.
\]
We will use the $L^2$ real scalar product defined by
$$(u,v)=\RE\int_{\RR^{n+d}}\overline{u(x,z)}v(x,z)dxdz$$
and we shall denote
$$(u,v)_{B^m}=\sum_{|\kappa|\leq m}\left(\pa^\kappa_x u,\pa^\kappa_x v\right)+\left(|x|^{m}u,|x|^{m}v\right)+\left(\Lambda^m_z u,\Lambda^m_zv\right).$$

\begin{remark}
	Theorem VIII.33 of \cite{reed1980methods} ensures that if $\Lambda$ is an essentially self-adjoint operator on the Hilbert space $H_1$ of domain $\mathcal{D}(\Lambda)$ and $H_2$ is another Hilbert space, then $\Lambda\otimes I$ is essentially self-adjoint on $H_1\otimes H_2$ with domain $\mathcal{D}(\Lambda\otimes I)= \mathcal{D}(\Lambda)\otimes H_2$; here $I$ is the identity of $H_2$.
\end{remark}
\begin{remark}
Let us stress that showing that the domain of the self-adjoint operator $\mathcal{H}_z^{m/2}$ of $L^2(\RR^d)$ is $B^m(\RR^d)$ is a delicate point. It can be shown \cite{helffer1984theorie} that the following norms $N_1^m$, and $N_2^m$ defined for $u\in\mathcal{C}^\infty_0(\RR^d)$ by
\bee
	N_1^m(u) &=& \|\Lambda^m_zu\|_{L^2(\RR^d)},\\
	N_2^m(u) &=& \|u\|_{H^m(\RR^d)}+\||z|^{m}u\|_{L^2(\RR^d)}
\eee
are equivalent. In the sequel, we will also make frequent use of the estimate
\[
	\||x|^k\partial_z^{\kappa}u\|_{L^2}\leq CN^m_1(u),\quad \mbox{for all }u\in \mathcal{C}^\infty_0(\RR^d) \mbox{ and }k+|\kappa|\leq m.
\]
Ben Abdallah et al. generalized these results for a more general class of confining potential using Weyl-H\"ormander calculus in \cite{BenAbdallah2008154}.
\end{remark}
As an immediate consequence, we get the following \textit{tame estimate}  for $m\geq0$. Let $G:\CC\rightarrow\CC$ be a smooth function such that $G(0)=0$, then for all $u\in B^m(\RR^{n+d})\cap L^\infty(\RR^{n+d})$, we have $G(u)\in B^m(\RR^{n+d})\cap L^\infty(\RR^{n+d})$ and
\be
\label{eq:tamez}
	\|\Lambda^m_zG(u)\|_{L^2}\leq C\left(\|u\|_{L^\infty}\right)\|\Lambda^m_z u\|_{L^2}
\ee
where $C:[0,\infty[\rightarrow [0,\infty[$ (see \cite[Proposition 2.5]{BenAbdallah2008154}, \cite[Lemma 4.10.2]{cazenave2003semilinear} or \cite[Lemma 1.24]{bookRemi}).

\begin{remark}\label{rem:algebraBm}
	Assuming that $m>\frac{n+d}{2}$, we get that
	\[
		B^m(\RR^{n+d})\hookrightarrow H^m(\RR^{n+d})\hookrightarrow L^\infty(\RR^{n+d}),
	\]
	 and $B^m(\RR^{n+d})$ is an algebra.
\end{remark}
The proof of uniform well-posedness for the four equations  \eqref{eq:b1}, \eqref{eq:b2}, \eqref{eq:b6} and \eqref{eq:b3} will be based on the following lemma concerning a non-homogeneous linear equation (\ref{eq:d}) with
a given source term $R$.
\begin{lemma}\label{lem:bound}
	Let us assume that for some $m>2$, $s\geq m+2$ and $T>0$, we have
	\begin{enumerate}[(i)]
		\item $a_0\in B^{m}(\RR^{n+d})$,
		\item $S\in \mathcal C([0,T];{\tt SQ}_s(\RR^n))\cap \mathcal{C}^{s}([0,T]\times\RR^n)$ solves the eikonal equation \eqref{eq:S},
		\item $R\in \mathcal{C}([0,T];B^{m}(\RR^{n+d}))$.
	\end{enumerate}
	Then, for all $\alpha\in [0,1]$, there exists a unique solution $a\in \mathcal{C}([0,T]; B^{m}(\RR^{n+d}))\cap  \mathcal{C}^1([0,T]; B^{m-2}(\RR^{n+d}))$ to the following equation:
	\begin{equation}\label{eq:d}
					\pa_t a +\nabla_x S \cdot \nabla_x a  +\frac{a}{2} \Delta_x S=i\frac{\alpha}{2}\Delta_x a +R,\quad			 a(0,x,z)=a_0(x,z).
	\end{equation}	
	Moreover for all $t\in[0,T]$, $a$ satisfies the estimates
	\begin{align}\label{eq:estd}
		 \left\|a(t)\right\|_{B^m}^2&\leq\left\|a_0\right\|_{B^m}^2+C\int_0^t\left\|a(s)\right\|_{B^m}^2ds+\int_0^t\left(a(s),R(s)\right)_{B^m}ds
		 \\\label{eq:estd2}&\leq\left\|a_0\right\|_{B^m}^2+C\int_0^t\left(\left\|a(s)\right\|_{B^m}^2+\left\|R(s)\right\|_{B^m}^2\right)ds
	\end{align}
	where $C$ is a generic constant which depends only on $m$ and on $$\underset{2\leq|\kappa|\leq s}{\sup} \|\pa^\kappa_x S\|_{L^\infty([0,T]\times \RR^n)}.$$
\end{lemma}
\begin{proof}
We first prove the result for $0<\alpha\leq 1$ and treat the case $\alpha=0$ in a second step. Let us start with a few preliminary remarks. From assumption {\em (ii)}, we deduce that $|\nabla_x S(t,x)|\leq C(1+|x|)$ and that, for all $2\leq k\leq m$, we have the equivalences
$$u\in \mathcal{C}([0,T]; B^{k}(\RR^{n+d}))\Longleftrightarrow ue^{iS/\alpha}\in \mathcal{C}([0,T]; B^{k}(\RR^{n+d}))
$$
and
$$u\in \mathcal{C}^1([0,T]; B^{k-2}(\RR^{n+d}))\Longleftrightarrow  ue^{iS/\alpha}\in\mathcal{C}^1([0,T]; B^{k-2}(\RR^{n+d})).
$$
Moreover, a direct calculation using the fact that $S$ solves the eikonal equation shows that a function $a\in \mathcal{C}([0,T]; B^{2}(\RR^{n+d}))\cap  \mathcal{C}^1([0,T]; L^2(\RR^{n+d}))$ is a strong solution of \eqref{eq:d} if, and only if $\Psi=ae^{iS/\alpha}$ is a strong solution of the non-homogeneous linear GPE
	\begin{equation}
	\label{GPElinear}
		i\alpha \pa_t \Psi=-\frac{\alpha^2}{2}\Delta_x \Psi+\frac{|x|^2}{2}\Psi+i\alpha R\,e^{iS/\alpha},\quad	 \Psi(0,x,z)=a_0(x,z)e^{iS_0(x)/\alpha}.
	\end{equation}
	Note that $R\,e^{iS/\alpha}\in \mathcal{C}([0,T]; B^{m}(\RR^{n+d}))$ and that $a_0e^{iS_0/\alpha}\in B^m$. Therefore, standard results on Schr\"odinger equations \cite{cazenave2003semilinear} give the existence and uniqueness of the strong solution $\Psi\in \mathcal{C}([0,T]; B^{m}(\RR^{n+d}))\cap  \mathcal{C}^1([0,T]; B^{m-2}(\RR^{n+d}))$ to \eqref{GPElinear}. This solution can be expressed in terms of the Duhamel formula
	$$\Psi(t)=e^{-it\mathcal H_x}\left(a_0e^{iS_0/\alpha}\right)+\int_0^te^{-i(t-\tau)\mathcal H_x}\left(R(\tau)e^{iS(\tau)/\alpha}\right)d\tau,$$
	where $\mathcal H_x=-\frac{\alpha^2}{2}\Delta_x+\frac{|x|^2}{2}$. This proves the well-posedness of \eqref{eq:d} for $0<\alpha \leq 1$.
	
	\bs
	
	Let us now prove the estimate \eqref{eq:estd}. Applying $\pa^\kappa_x$ to equation \eqref{eq:d}, where $|\kappa|\leq m$, yields
	$$
		\pa_t \pa^\kappa_x a + \nabla_xS\cdot\nabla_x \pa^\kappa_xÊa-[\nabla_xS\cdot\nabla_x,\pa^\kappa_x]a+\pa^\kappa_x\left(\frac{a}{2}\Delta_x S\right)=i\frac{\alpha}{2}\Delta_x \pa^\kappa_xa+\pa^\kappa_xR.
	$$
Take the $L^2$ real scalar product of this equation with $\pa^\kappa_x a$. Since $i\Delta_x$ is skew-symmetric, we get  that
	$$
		 \frac{1}{2}\frac{d}{dt}\|\pa^\kappa_xa\|_{L^2}^2+\left(\pa^\kappa_xa,\nabla_xS\cdot\nabla_x\pa^\kappa_xa\right)=\left(\pa^\kappa_x a, R_1+R_2+\pa^\kappa_xR\right),
	$$
	where $R_1 = [\nabla_x S\cdot \nabla_x,\pa^\kappa_x]a$ and $R_2=- \pa^\kappa_x\left(\frac{a}{2}\Delta_x S\right)$. We have by an integration by part that
	\[
		\left|\left(\pa^\kappa_xa,\nabla_xS\cdot\nabla_x\pa^\kappa_x a\right)\right| =
		 \frac{1}{2}\left|\left(\Delta_xS\,\pa^\kappa_xa,\pa^\kappa_xa\right)\right|.
	\]
	 We recall that $S\in \mathcal C([0,T],{\tt SQ}_s)$ with $s\geq m+2$, hence all the derivatives of $\Delta_x S$ up to the order $m$ are bounded, so that
	 \[
		\left|\left(\pa^\kappa_xa,\nabla_xS\cdot\nabla_x\pa^\kappa_x a\right)\right|
		\leq C\|a\|^2_{B^m}
	\]
	and
	\[
		\|R_2\|_{L^2}\leq C\|a\|_{B^m}.
	\]
	Let us now remark that the commutator $[\nabla_x S\cdot \nabla_x,\pa^\kappa_x]a$ is only composed of differential operators of order $\leq m$ multiplied by $L^\infty$ functions, since $S$ is subquadratic.
	Hence, we get
	$$
		\|R_1\|_{L^2}\leq C\|a\|_{B^m}.
	$$
	It comes finally
	\begin{equation}
	\label{esti1}
		\frac{1}{2}\frac{d}{dt}\|\pa^\kappa_xa\|_{L^2}^2\leq C\|a\|_{B^m}^2+\left(\pa^\kappa_x a, \pa^\kappa_xR\right).
	\end{equation}
	
	Applying now the operator $\Lambda^m_z$ to \eqref{eq:d} yields (recall that $S$ does not depend on $z$)
	$$
		\pa_t \Lambda^m_z a + \nabla_xS\cdot\nabla_x \Lambda^m_zÊa+\Lambda^m_z\left(\frac{a}{2}\Delta_x S\right)=i\frac{\alpha}{2}\Delta_x \Lambda^m_za+\Lambda^m_zR.
	$$
	Hence, taking the $L^2$ real scalar product with $\Lambda^m_za $ gives, after integrations by parts,
	\begin{align}\label{esti2}
		\frac{1}{2}\frac{d}{dt}\|\Lambda^m_za\|_{L^2}^2&=\left(\Lambda^m_z a,\pa^\kappa_xR\right).
			\end{align}

	Let us finally apply the operator $|x|^m$ to \eqref{eq:d}. We get that
	\begin{align}\nonumber
		\pa_t \left(|x|^ma\right)+\nabla_xS\cdot\nabla_x(|x|^ma)+\frac{|x|^ma}{2}\Delta_x S=&[\nabla_xS\cdot \nabla_x,|x|^m]a+\frac{i\alpha}{2}\Delta_x\left(|x|^ma\right)\\
	&+\frac{i\alpha}{2}[|x|^m,\Delta_x]a+|x|^mR.\label{eqxm}
	\end{align}
	Since $S$ is subquadratic, we have
	\[
		\|[\nabla_xS\cdot \nabla_x,|x|^m]a\|_{L^2}=m\| |x|^{m-2}x\cdot \nabla_xS\,a\|_{L^2}\leq C\|\left(1+|x|^m\right)a\|_{L^2}\leq C\|a\|_{B^m}
	\]
	and we compute also
	\begin{align*}
		\|[|x|^m,\Delta_x]a\|_{L^2}&=\|\Delta_x(|x|^m)a+2\nabla_x(|x|^m)\cdot \nabla_x a\|_{L^2}\\
		&\leq C\||x|^{m-2}a\|_{L^2}+C\||x|^{m-1}\na_x a\|_{L^2}\leq C\|a\|_{B^m}.
	\end{align*}
	Taking the $L^2$ real scalar product of \eqref{eqxm} with $|x|^ma$, we get
	\begin{equation}
	\label{esti3}
		\frac{1}{2}\frac{d}{dt}\||x|^ma\|_{L^2}^2\leq C \|a\|_{B^m}^2+\left(|x|^m a,|x|^mR\right).
		\end{equation}
Finally, from \eqref{esti1}, \eqref{esti2} and \eqref{esti3}, we deduce \eqref{eq:estd} for $0<\alpha\leq 1$. From \eqref{eq:estd} and Cauchy-Schwarz, we obtain then the second estimate \eqref{eq:estd2}. Note that the above calculations are rigorous only if we know a better regularity for $a$, for instance $a\in \mathcal{C}([0,T]; B^{m+2}(\RR^{n+d}))\cap  \mathcal{C}^1([0,T]; B^{m}(\RR^{n+d}))$. A standard regularization argument, that we skip here, enables to fully justify this proof.


\bs
Let us now prove the result in the case $\alpha=0$. To this aim, we consider a regularized sequence $a_0^\delta$, $S^\delta$, $R^\delta$, where $\delta>0$ is a regularization parameter, such that
\begin{enumerate}[(i)']
		\item $a_0^\delta\in B^{m+2}(\RR^{n+d})$ and $\|a_0^\delta-a_0\|_{B^m}\to 0$ as $\delta\to 0$,
		\item $S^\delta\in \mathcal C([0,T^\delta];{\tt SQ}_{s+2}(\RR^n))\cap \mathcal{C}^{s+2}([0,T^\delta]\times\RR^n)$ solves the eikonal equation \eqref{eq:S}, $\|S^\delta-S\|_{\mathcal C^s}\to 0$ as $\delta\to 0$ and $T^\delta\to T$,
		\item $R^\delta\in \mathcal{C}([0,T];B^{m+2}(\RR^{n+d}))$ and $\|R^\delta-R\|_{L^\infty([0,T];B^m)}\to 0$ as $\delta\to 0$.
	\end{enumerate}
	Note that, to construct $S^\delta$, we need to regularize the associated initial data $S_0$, which may make the existence time $T$ depend on $\delta$.

We consider a sequence  $(\alpha_n)_{n\in \NN}$ of positive numbers converging to $0$ and denote by $(a_n^\delta)_{n\in \NN}$ the sequence of solutions of
$$
					\pa_t a_n^\delta +\nabla_x S^\delta \cdot \nabla_x a_n^\delta  +\frac{a_n^\delta}{2} \Delta_x S^\delta=i\frac{\alpha_n}{2}\Delta_x a_n^\delta +R^\delta,\quad			a_n^\delta(0,x,z)=a^\delta_0(x,z).
$$
	In a first step, we consider $\delta>0$ as fixed. From \eqref{eq:estd2} and Gronwall's lemma, we infer
$$
		\max_{0\leq t\leq T^\delta}\left\|a^\delta_n(t)\right\|_{B^{m+2}}^2\leq\left(\left\|a^\delta_0\right\|_{B^{m+2}}^2+CT^\delta\|R^\delta\|^2_{L^\infty([0,T^\delta];B^{m+2})}\right)e^{CT^\delta},
$$
so this sequence $(a_n^\delta)_{n\in \NN}$ is uniformly bounded in $\mathcal C([0,T^\delta];B^{m+2}(\RR^{n+d}))$.
	Moreover, we have
	\begin{equation}
	\label{eq:estdp}
		\pa_t (a_p^\delta-a_q^\delta) +\nabla_x S^\delta \cdot \nabla_x  (a_p^\delta-a_q^\delta)  +\frac{ (a_p^\delta-a_q^\delta)}{2} \Delta_x S^\delta=i\frac{\alpha_q}{2}\Delta_x (a_p^\delta-a_q^\delta)+ i\frac{\alpha_p-\alpha_q}{2}\Delta_x a_p^\delta.
	\end{equation}
	Applying again \eqref{eq:estd2} with $R$ replaced by $i\frac{\alpha_p-\alpha_q}{2}\Delta a_p^\delta\in \mathcal{C}([0,T^\delta];B^m(\RR^{n+d})) $ and $a_0=0$ gives
	\[
		\|a_p^\delta(t)-a_q^\delta(t)\|_{B^{m}}^2\leq C_\delta(\alpha_p-\alpha_q)^2+C_\delta\int_0^t\left(\|a_p^\delta(s)-a_q^\delta(s)\|_{B^{m}}^2\right)ds
	\]
	and Gronwall's lemma implies that
	\[
		\max_{0\leq t\leq T^\delta}\|a_p^\delta(t)-a_q^\delta(t)\|_{B^{m}}^2\leq C_\delta(\alpha_p-\alpha_q)^2.
	\]
	Hence, $(a_n^\delta)_{n\in \NN}$ is a Cauchy sequence of $\mathcal{C}([0,T^\delta];B^m(\RR^{n+d}))$. Inserting this information in \eqref{eq:estdp} yields that it is also a Cauchy sequence of $\mathcal{C}^1([0,T^\delta];B^{m-2}(\RR^{n+d}))$. Therefore, as $n\to +\infty$, this sequence converges  to a function $$a^\delta \in \mathcal{C}([0,T^\delta];B^m(\RR^{n+d}))\cap \mathcal{C}^1([0,T^\delta];B^{m-2}(\RR^{n+d}))$$ which solves
	\begin{equation}\label{eq:deps}
					\pa_t a^\delta +\nabla_x S^\delta \cdot \nabla_x a^\delta  +\frac{a^\delta}{2} \Delta_x S^\delta= R^\delta,\quad			a^\delta(0,x,z)=a_0^\delta(x,z).
	\end{equation}	
	
	Let us now proceed to the limit $\delta\to 0$. Using \eqref{eq:estd2} for \eqref{eq:deps} (remark that the above proof of this estimate is valid also for $\alpha=0$) enables to show that $a^\delta$ is a Cauchy sequence in $\mathcal{C}([0,T];B^m(\RR^{n+d}))\cap \mathcal{C}^1([0,T];B^{m-2}(\RR^{n+d}))$ and converges to a function $a$ which satisfies \eqref{eq:d} with $\alpha=0$. The estimates \eqref{eq:estd} and \eqref{eq:estd2} are also valid for this function $a$. Remark that the uniqueness of the solution $a$ also stems from the estimate \eqref{eq:estd2} (written for the difference between two solutions) and Gronwall's lemma.
\end{proof}
In order to prove the uniform well-posedness of the four nonlinear equations \eqref{eq:b1}, \eqref{eq:b2}, \eqref{eq:b3} and \eqref{eq:b6}, we will need the following Lipschitz estimates for $g:u\mapsto |u|^2u$,  $F(\theta,\cdot)$ defined by \eqref{def:F} and $F_{av}$ defined by \eqref{def:Fav}.
\begin{lemma}\label{lem:tame}
For all $m>\frac{n+d}{2}$ and  $M>0$, there is a nondecreasing function $M\mapsto C_m(M)>0$ such that
\bee
	\|g(u)-g(v)\|_{B^m}&\leq C_m(M)\|u-v\|_{B^m}\\
	\|F_{av}(u)-F_{av}(v)\|_{B^m}&\leq C_m(M)\|u-v\|_{B^m}\\
	\left\|F\left(\theta,u\right)-F\left(\theta,v\right)\right\|_{B^m}&\leq C_m(M)\|u-v\|_{B^m},
\eee
for all $u,v\in B^m(\RR^{n+d})$ satisfying $\|u\|_{B^m}\leq M$, $\|v\|_{B^m}\leq M$ and for all $\theta\in \RR$.
\end{lemma}
\begin{proof}
	The first inequality is already given in \cite{BenAbdallah2008154,helffer1984theorie}. Using that $e^{-i\theta \mathcal{H}_z}$ is an isometry on $B^m$, we get that
	\begin{align*}
		\left\|F\left(\theta,u\right)-F\left(\theta,v\right)\right\|_{B^m}&= \left\|g\left(e^{-i\theta \mathcal{H}_z}u\right)-g\left(e^{-i\theta \mathcal{H}_z}v\right)\right\|_{B^m}\\
		&\leq C_m(M)\left\|e^{-i\theta \mathcal{H}_z}u-e^{-i\theta \mathcal{H}_z}v\right\|_{B^m}\\
		&\leq C_m(M)\left\|u-v\right\|_{B^m}
	\end{align*}
	and
	\begin{align*}
		\left\|F_{av}\left(u\right)-F_{av}\left(v\right)\right\|_{B^m}&= \left\|\frac{1}{2\pi}\int_0^{2\pi}\left(F\left(\theta,u\right)-F\left(\theta,v\right)\right)d\theta\right\|_{B^m}\\
		&\leq \frac{1}{2\pi}\int_0^{2\pi}\left\|F\left(\theta,u\right)-F\left(\theta,v\right)\right\|_{B^m}d\theta\\
		&\leq C_m(M)\left\|u-v\right\|_{B^m}.
	\end{align*}
\end{proof}
The main result of this section is the following Proposition.
\begin{proposition}\label{lem:existbound}
	Let $(\eps,\alpha)\in (0,1]^2$ and $M>0$. Let $m>\frac{n+d}{2}$ and $s\geq m+2$. Let $S_0\in {\tt SQ}_{s+1}(\RR^n)$ and $S$ be the corresponding solution of the eikonal equation, given by Proposition \pref{lem:eik}. Then there exist $\overline T\in (0,T]$ which depends only on $M$ and
	\[
		\underset{2\leq|\kappa|\leq s}{\sup} \|\pa^\kappa_x S\|_{L^\infty([0,T]\times \RR^n)}
	\]
	such that, for all $A_0\in B^m(\RR^{n+d})$ satisfying $\|A_0\|_{B^m(\RR^{n+d})}\leq M$,
	\begin{enumerate}[(i)]
		\item\label{lemexistence:pt1} there is a unique solution $\beq\in \mathcal{C}([0,\overline T];B^m(\RR^{n+d}))\cap \mathcal{C}^1([0,\overline T];B^{m-2}(\RR^{n+d}))$ to Eq. \eqref{eq:b1},
		\item  there is a unique solution $\beqo\in \mathcal{C}([0,\overline T];B^m(\RR^{n+d}))\cap \mathcal{C}^1([0,\overline T];B^{m-2}(\RR^{n+d}))$ to Eq. \eqref{eq:b2},
		\item  there is a unique solution $\beqs\in \mathcal{C}([0,\overline T];B^m(\RR^{n+d}))\cap \mathcal{C}^1([0,\overline T];B^{m-2}(\RR^{n+d}))$ to Eq. \eqref{eq:b6},
		\item  there is a unique solution $\beqoo\in \mathcal{C}([0,\overline T];B^m(\RR^{n+d}))\cap \mathcal{C}^1([0,\overline T];B^{m-2}(\RR^{n+d}))$ to Eq. \eqref{eq:b3}.
	\end{enumerate}
	Moreover, we have
		\[
			\|\beq\|_{L^\infty([0,\overline T];B^m)},\;\|\beqo\|_{L^\infty([0,\overline T];B^m)},\;\|\beqs\|_{L^\infty([0,\overline T];B^m)},\;\|\beqoo\|_{L^\infty([0,\overline T];B^m)}\leq 2M,
		\]
		and the $L^\infty([0,\overline T];B^{m-2})$ norms of
		$\pa_t\beq$, $\pa_t\beqo$, $\pa_t\beqs$ and  $\pa_t\beqoo$ are uniformly bounded with respect to $(\eps,\alpha)$.
\end{proposition}
\begin{proof}
	This Proposition can be proved by iterative schemes. Let us only write the proof of Item {\em (\ref{lemexistence:pt1})}, the other items can be proved similarly. We denote by $a^0$ the function defined for all $t\in [0,T]$ by $a^0(t)=A_0$. Then, for all $k\in \NN$, $a^{k+1}$ is defined as the solution of the following equation
	\begin{equation*}
					\pa_t a^{k+1} +\nabla_x S \cdot \nabla_x a^{k+1} +\frac{a^{k+1}}{2} \Delta_x S=i\frac{\alpha}{2}\Delta_x a^{k+1} -iF\left(\frac{t}{\eps^2},a^{k}\right),
	\end{equation*}	
	satisfying
	\[
		a^{k+1}(0,x,z)=A_0(x,z).
	\]
	From Lemmas \ref{lem:bound} and \ref{lem:tame}, we deduce that the sequence $(a^k)_{k\in \NN}$ is well-defined in $ \mathcal{C}([0,T]; B^{m}(\RR^{n+d}))\cap  \mathcal{C}^1([0,T]; B^{m-2}(\RR^{n+d}))$ and that
	\begin{align*}
		\left\|a^{k+1}(t)\right\|_{B^m}^2
		&\leq \left\|A_0\right\|_{B^m}^2+C_0\int_0^t\left(\left\|a^{k+1}(s)\right\|_{B^m}^2+\left\|F\left(\frac{s}{\eps^2},a^{k}\right)\right\|_{B^m}^2\right)ds\\
		&\leq \left\|A_0\right\|_{B^m}^2+C_0\int_0^t\left(\left\|a^{k+1}(s)\right\|_{B^m}^2+C_m(\|a^{k}\|_{B^m})^2\|a^{k}\|_{B^m}^2\right)ds,
	\end{align*}
	where we used that $F(\theta,0)=0$. Let us prove by induction that, for
	\[
		\overline T=\min\left(T,\frac{\log 2}{C_0},\,\sup\left\{ t>0:\;tC_0C_m(2M)^2e^{C_0t} \leq \frac{1}{2}\right\}\right),
	\]
		 we have
	\[
		\max_{0\leq t\leq \overline T}\|a^k(t)\|_{B^m}\leq 2M
	\]
	for all $k\in \NN$. This property is clearly true for $k=0$. Assume that this condition is satisfied for $k\in \NN$. By Gronwall's lemma, we obtain that this property is also true for $k+1$, since
	\[
		\left\|a^{k+1}(t)\right\|_{B^m}^2\leq e^{C_0t}M^2+ tC_0C_m(2M)^2e^{C_0t}4M^2\leq 4M^2\quad \mbox{for }0\leq t\leq \overline{T}.
	\]
	Now, for all $k\in\NN^*$, we get by Lemmas \ref{lem:bound} and \ref{lem:tame} that
	\begin{align*}
		&\|a^{k+1}(t)-a^k(t)\|_{B^m}^2\\
		&\qquad \leq C_0\int_0^t\left(\|a^{k+1}(s)-a^k(s)\|_{B^m}^2+C_m(2M)^2\|a^{k}(s)-a^{k-1}(s)\|_{B^m}^2\right)ds
	\end{align*}
	and Gronwall's lemma ensures that
	$$
		\|a^{k+1}(t)-a^k(t)\|_{B^m}^2\leq C_0C_m(2M)^2\int_0^t\left(e^{C_0(t-s)}\|a^{k}(s)-a^{k-1}(s)\|_{B^m}^2\right)ds.
	$$
Then we obtain for $t\in[0,\overline T]$ that
	\bee
		\lefteqn{\|a^{k+1}-a^k\|_{L^\infty([0,\overline T];B^m)}^2}\\
		&&\leq \overline  TC_0C_m(2M)^2e^{C_0\overline T}\|a^{k}-a^{k-1}\|_{L^\infty([0,\overline  T];B^m)}^2\leq 2^{-k}\|a^{1}-a^{0}\|_{L^\infty([0,\overline  T];B^m)}^2.
	\eee	
	 Hence the series $(a^{k+1}-a^{k})_{k\in \NN}$ converges in $\mathcal{C}([0,\overline  T];B^m(\RR^{n+d}))$ so that $(a^k)_{k\in \NN}$ converges to a solution $\beq$ of equation \eqref{eq:b1}. Let us remark that $\beq$  satisfies the uniform estimate
	\[
		\|\beq\|_{L^\infty([0,\overline  T];B^m)}\leq 2M.
	\]
		Inserting this estimate into \eqref{eq:b1} and using that $S$ is subquadratic yields a uniform estimate of $\|\partial_t \beq\|_{L^\infty([0,\overline  T];B^{m-2})}$. The uniqueness property follows also from Gronwall's lemma and from Lemma \ref{lem:tame}.
\end{proof}
%
%
%
%
%
\section{Proof of Theorem \ref{theo:limit}: the limits $\alpha\to 0$ and $\eps\to 0$}\label{sec3}

This section is devoted to the proof of Theorem \ref{theo:limit}.

\ms
\ni
	\textit{Strong confinement limits: proof of \eqref{error1} and \eqref{error2}.}
	 Let us introduce the function
	\[
		\begin{array}{llll}
			\mathscr{F}:	&\RR\times B^m(\RR^{n+d})&\longrightarrow & B^m(\RR^{n+d})\\
				&(\theta,u)&\longmapsto&\int_0^\theta (F(s,u)-F_{av}(u))ds.
		\end{array}
	\]
	which satisfies the following properties for every $u\in B^m(\RR^{n+d})$:
	\begin{enumerate}
		\item[(a)] $\theta\mapsto \mathscr{F}(\theta,u)$ is a $2\pi$-periodic function, since $\theta\mapsto F(\theta,u)$ is $2\pi$-periodic and $F_{av}$ is its average,
		\item[(b)] 	\label{eq:lemboundH1}
		if $\|u\|_{B^m}\leq M$ then
		$\|\mathscr{F}(\theta,u)\|_{B^m}\leq 4\pi C_m(M)M$ for all $\theta\in\RR$,
		where $C_m(\cdot)$ was defined in Lemma \ref{lem:tame}.
	\end{enumerate}
	 Using the relation
	\[
		\eps^2\frac{d}{dt}\left(\mathscr{F}(t/\eps^2,u(t))\right) = \left(F(t/\eps^2,u(t))-F_{av}(u(t))\right)+\eps^2 D_u\mathscr{F}(t/\eps^2,u(t))(\pa_s u(t))
	\]
	 and equations \eqref{eq:b2} and \eqref{eq:b1} (or their versions with $\alpha=0$, i.e. \eqref{eq:b3} and \eqref{eq:b6}), we obtain for all $\alpha\in[0,1]$ and $\eps\in (0,1]$,
	\bea \label{eq:lemdiff1}
		\left(\pa_t + \nabla_x S\cdot \nabla_x +\frac{\Delta_x S }{2}-\frac{i\alpha}{2}\Delta_x\right) \left(\beq-\beqo\right) =-i\left(F_{av}(\beq)-F_{av}(\beqo)\right) \nonumber\\
		 -i\eps^2 \pa_t\mathscr{F}(t/\eps^2,\beq)+i\eps^2D_u\mathscr{F}(t/\eps^2,\beq)(\pa_t \beq).\quad
	\eea	
	Hence, Lemma \ref{lem:bound} ensures that
	\bee
		&&\|\beq(t)-\beqo(t)\|^2_{B^{m-2}}\leq C \int_0^t\|\beq(s)-\beqo(s)\|^2_{B^{m-2}}ds + I_1+I_2+I_3 +I_4\\
	\eee
	where
	\bee
		&&I_1 = \int_0^t\left(\|F_{av}(\beq)-F_{av}(\beqo)\|^2_{B^{m-2}}+\eps^4\|D_u \mathscr{F}(s/\eps^2,\beq)(\pa_t\beq)\|^2_{B^{m-2}}\right)ds,\\
		&&I_2 = \sum_{|\kappa|\leq m-2}\eps^2\int_0^t\left(\pa^\kappa_x \left(\beq-\beqo\right), -i\frac{d}{ds}\pa^\kappa_x\mathscr{F}(s/\eps^2,\beq)\right)ds,\\
		&&I_3 =  \eps^2\int_0^t\left(\Lambda^{m-2}_z \left(\beq-\beqo\right),-i\frac{d}{ds}\Lambda^{m-2}_z\mathscr{F}(s/\eps^2,\beq)\right)ds,\\
		&&I_4 = \eps^2\int_0^t\left(|x|^{2({m-2})}\left(\beq-\beqo\right),-i\frac{d}{ds}\mathscr{F}(s/\eps^2,\beq)\right)ds.
	\eee
	
	Let us remark that according to Theorem \ref{theo:local1} {\em (iii)}, the sequences  $(\beq)_{\eps,\alpha}$ and $(\pa_t \beq)_{\eps,\alpha}$ are uniformly bounded, respectively in $L^\infty([0,\overline T];B^m(\RR^{n+d}))$ and in $L^\infty([0,\overline T];B^{m-2}(\RR^{n+d}))$. Moreover, since $m-2>\frac{n+d}{2}$, Remark \ref{rem:algebraBm} ensures that $B^{m-2}(\RR^{n+d})$ is an algebra, so that, for all $\theta\in[0,2\pi]$,
	\begin{align*}
		&\|D_u\mathscr{F}(\theta,\beq)(\pa_t \beq)\|_{B^{m-2}}\\
   &\ = \left\|\int_0^\theta \left(D_u F(s,\beq)-D_u F_{av}(\beq)\right)(\pa_t \beq) ds\right\|_{B^{m-2}}\\
		 &\ \leq \int_0^\theta\left\|2|e^{-is\mathcal{H}_z}\beq|^2e^{-is\mathcal{H}_z}\pa_t \beq +(e^{-is\mathcal{H}_z}\beq)^2e^{is\mathcal{H}_z}\overline{\pa_t \beq}\right\|_{B^{m-2}}ds\\
		&\quad+\frac{\theta}{2\pi}\int_0^{2\pi}\|2|e^{-is\mathcal{H}_z}\beq|^2e^{-is\mathcal{H}_z}\pa_t \beq +(e^{-is\mathcal{H}_z}\beq)^2e^{is\mathcal{H}_z}\overline{\pa_t \beq} \|_{B^{m-2}}ds
	\end{align*}
	satisfies
	\bea \label{lem:ineq1}
		\sup_{s\in[0,\overline T]}\sup_{\eps,\alpha}\|D_u\mathscr{F}(s/\eps^2,\beq(s))(\pa_t \beq(s))\|_{B^{m-2}}\leq C.
	\eea
	Using inequality \eqref{lem:ineq1} and Lemma \ref{lem:tame}, we obtain that
	\[
		I_1 \leq C \int_0^t\|\beq(s)-\beqo(s)\|^2_{B^{m-2}}ds + C\eps^4.
	\]
	Let us study the three remaining terms $I_2$, $I_3$ and $I_4$. By an integration by parts, we get that
	 \bee
	 	I_2 &=&  \sum_{|\kappa|\leq m-2}\eps^2\left(\pa^\kappa_x \left(\beq(t)-\beqo(t)\right),-i\pa^\kappa_x\mathscr{F}(t/\eps^2,\beq)\right)+ \sum_{|\kappa|\leq m-2}J^\kappa\\
		&\leq& C\eps^4+ \frac{1}{4}\|\beq(t)-\beqo(t)\|^2_{B^{{m-2}}}+ \sum_{|\kappa|\leq m-2}J^\kappa
	 \eee
	 where
	 $$
	 	J^\kappa = -\eps^2\int_0^t\left(\pa^\kappa_x \pa_t\left(\beq-\beqo\right),-i\pa^\kappa_x\mathscr{F}(s/\eps^2,\beq)\right)ds.
	 $$
	 Using again equation \eqref{eq:lemdiff1}, we get that $J^\kappa= J^\kappa_1+ J^\kappa_2 + J^\kappa_3 + J^\kappa_4 + J^\kappa_5$, where
	 \bee
	 	J^\kappa_1 &=& \eps^2\int_0^t\left(\pa^\kappa_x \left(\nabla_xS\cdot\nabla_x\right)\left(\beq-\beqo\right),-i\pa^\kappa_x\mathscr{F}(s/\eps^2,\beq)ds\right),\\
		J^\kappa_2 &=& \eps^2\int_0^t\left(\pa^\kappa_x \left( \frac{\Delta_x S}{2}-\frac{i\alpha}{2}\Delta_x\right)\left(\beq-\beqo\right), -i\pa^\kappa_x\mathscr{F}(s/\eps^2,\beq)\right)ds,\\
		J^\kappa_3 &=& -\eps^2\int_0^t\left(\pa^\kappa_x \left(F_{av}(\beq)-F_{av}(\beqo)\right),\pa^\kappa_x\mathscr{F}(s/\eps^2,\beq)\right)ds,\\
		J^\kappa_4 &=& \eps^4\int_0^t\left(\pa^\kappa_x D_u \mathscr{F}(s/\eps^2,\beq)(\pa_t \beq),\pa^\kappa_x\mathscr{F}(s/\eps^2,\beq) \right)ds,\\
		J^\kappa_5 &= &-\eps^4\int_0^t\left(\pa^\kappa_x\frac{d}{ds}\mathscr{F}(s/\eps^2,\beq),\pa^\kappa_x\mathscr{F}(s/\eps^2,\beq)\right)ds.
	 \eee
	 Using an integration by parts and the fact that the commutator $[\pa^\kappa_x,  \nabla_xS\cdot\nabla_x]$ is an operator of order $m-2$, we get that
	 \bee
	  J_1^\kappa &=& \eps^2 \int_0^t \left([\pa^\kappa_x,  \nabla_xS\cdot\nabla_x]\left(\beq-\beqo\right),-i\pa^\kappa_x\mathscr{F}(s/\eps^2,\beq)\right)ds\\
	  && + \eps^2 \int_0^t \left(\left(\nabla_xS\cdot\nabla_x\right) \pa^\kappa_x \left(\beq-\beqo\right),-i\pa^\kappa_x\mathscr{F}(s/\eps^2,\beq)\right)ds\\
	  &\leq& C\eps^2 \int_0^t \|\beq-\beqo\|_{B^{m-2}}\|\mathscr{F}(s/\eps^2,\beq)\|_{B^{m-2}}ds \\
	  && - \eps^2\int_0^t \left(\pa^\kappa_x \left(\beq-\beqo\right),-i\nabla_x \cdot \left(\nabla_x S\pa^\kappa_x\mathscr{F}(s/\eps^2,\beq)\right)\right)ds\\
	  &\leq& C\eps^2 \int_0^t \|\beq-\beqo\|_{B^{m-2}}\|\mathscr{F}(s/\eps^2,\beq)\|_{B^{m}}ds\\
	  &\leq& C\eps^4 + \int_0^t\|\beq-\beqo\|^2_{B^{{m-2}}}ds.
	 \eee
	 Since the operator $\frac{\alpha}{2}\Delta_x$ is symmetric, we obtain
	 \bee
	 	J_2^\kappa  &\leq&  C\eps^2 \int_0^t \|\beq-\beqo\|_{B^{m-2}}\|\mathscr{F}(s/\eps^2,\beq)\|_{B^{m-2}}ds\\
			&&+\eps^2\int_0^t\left(\pa^\kappa_x \left(\beq-\beqo\right),\pa^\kappa_x\frac{\alpha\Delta_x}{2}\mathscr{F}(s/\eps^2,\beq)\right)ds\\
			&\leq&C\eps^2 \int_0^t \|\beq-\beqo\|_{B^{m-2}}\|\mathscr{F}(s/\eps^2,\beq)\|_{B^{m}}ds\\
			&\leq& C\eps^4 + \int_0^t\|\beq-\beqo\|^2_{B^{{m-2}}}ds.
	 \eee
	 From the Lipschitz estimates of Lemma \ref{lem:tame}, we deduce also
	 \bee
	 	J_3^\kappa &\leq& C\eps^4 + \int_0^t\|\beq-\beqo\|^2_{B^{{m-2}}}ds
	 \eee
	 and inequality \eqref{lem:ineq1} implies $J_4^\kappa \leq C\eps^4$. Finally, we obtain that
	 $$
	 	J_5^\kappa  = -\frac{\eps^4}{2}\int_0^t\frac{d}{ds}\left\|\pa^\kappa_x\mathscr{F}(s/\eps^2,\beq)\right\|_{L^2}^2ds= -\frac{\eps^4}{2}\left\|\pa^\kappa_x\mathscr{F}(t/\eps^2,\beq)\right\|_{L^2}^2\leq 0.
	 $$
	 and, finally,
	 \[
	 	I_2 \leq C\eps^4 + C\int_0^t\|\beq(s)-\beqo(s)\|^2_{B^{m-2}}ds + \frac{1}{4}\|\beq(t)-\beqo(t)\|^2_{B^{m-2}}.
	 \]
	 	The same arguments hold for the two remaining terms $I_3$ and $I_4$. Hence, we obtain
	 \[
	 	\|\beq(t)-\beqo(t)\|^2_{B^{m-2}}\leq C\int_0^t\|\beq(s)-\beqo(s)\|^2_{B^{m-2}}ds+ C\eps^{4}
	 \]
	 so that, by Gronwall's lemma, we get \eqref{error1},
	 \[
	 	\|\beq-\beqo\|_{L^\infty([0,\overline{T}];B^{m-2})}\leq C\eps^{2} \quad \mbox{for all }\alpha\in (0,1],
	 \]
	 and \eqref{error2},
	  \[
	 	\|\beqs-\beqoo\|_{L^\infty([0,\overline{T}];B^{m-2})}\leq C\eps^{2}.
	 \]
	 \ms

	 \ni
	\subsubsection*{The semi-classical limits: proof of \eqref{error3} and \eqref{error4}.}
		The error estimates \eqref{error3} and \eqref{error4} are simple consequences of the uniform bounds given by Theorem \ref{theo:local1} {\em (iii)}. We have indeed
		\bee
			\lefteqn{\pa_t(\beq-\beqs)+\nabla_x S\cdot \nabla_x (\beq-\beqs)+\frac{\Delta_x S}{2}(\beq-\beqs)}\\
			&& = i\frac{\alpha}{2}\Delta_x (\beq-\beqs)-i(F(s/\eps^2,\beq)-F(s/\eps^2,\beqs))-i\frac{\alpha}{2}\Delta_x \beqs
		\eee
		so that, by Lemmas \ref{lem:bound} and \ref{lem:tame} and by the uniform bound for $\|\beqs\|_{L^\infty([0,\overline T];B^m)}$,
		\bee
			\lefteqn{\|\beq(t)-\beqo(t)\|^2_{B^{m-2}}\leq C\int_0^t\|\beq(s)-\beqs(s)\|^2_{B^{m-2}}ds}\\
			&&+\int_0^t\left(\|F(s/\eps^2,\beq(s))-F(s/\eps^2,\beqo(s))\|^2_{B^{m-2}}ds+\frac{\alpha}{2}\| \beqs(s)\|^2_{B^m}\right)ds\\
			&&\leq C\alpha+C\int_0^t\|\beq(s)-\beqs(s)\|^2_{B^{m-2}}ds
		\eee and by Gronwall's lemma
		\[
			\|\beq-\beqo\|^2_{\mathcal{C}([0,\overline T];B^{m-2})}\leq C\alpha.
		\]
		The proof of \eqref{error3} is complete. The same proof holds for \eqref{error4}, replacing the function $F$ by $F_{av}$. The proof of Theorem \ref{theo:limit} is complete.

	\qed
\subsection*{Acknowledgment}
\ni
This work was supported by the Singapore A*STAR SERC  PSF-Grant 1321202067 (W.B.)
and by the ANR-FWF Project Lodiquas ANR-11-IS01-0003 (L.L.T. and F.M.).

\bibliographystyle{siam}

\bigskip

 \end{document}